\documentclass[conference]{IEEEtran}[10pt]
\usepackage{amsmath}
\usepackage{amssymb}
\usepackage{graphicx}
\usepackage{subfigure}
\usepackage{a4wide}
\usepackage{algorithm}
\usepackage{algorithmic}
\usepackage{amsfonts}

\usepackage[utf8]{inputenc}
\usepackage{amsfonts}
\usepackage{amsthm}
\usepackage{color}
\usepackage{graphicx}
\usepackage{verbatim}

\newtheorem{theorem}{Theorem}

\newcommand{\citex}{\textsc{Citex}}
\newcommand{\cA}{\mathcal{A}}
\newcommand{\cP}{\mathcal{P}}
\newcommand{\xx}{\mathbf{x}}
\newcommand{\yy}{\mathbf{y}}

\begin{document}
\title{{\citex}: A new citation index to measure the relative importance of authors and papers in scientific publications}
\author{Arindam Pal $\dag$ and Sushmita Ruj $\ddag$\\\\
$\dag$ TCS Innovation Labs,
Kolkata, India.
Email address: arindamp@gmail.com\\
$\ddag$ Indian Statistical Institute,
Kolkata, India.
Email address: sush@isical.ac.in}
\date{}
\maketitle

\begin{abstract}
Evaluating the performance of researchers and measuring the impact of papers written by scientists is the main objective of citation analysis. Various indices and metrics have been proposed for this. In this paper, we propose a new citation index {\citex}, which gives normalized scores to authors and papers to determine their rankings. To the best of our knowledge, this is the first citation index which simultaneously assigns scores to both authors and papers. Using these scores, we can get an objective measure of the reputation of an author and the impact of a paper.

We model this problem as an iterative computation on a publication graph, whose vertices are authors and papers, and whose edges indicate which author has written which paper. We prove that this iterative computation converges in the limit, by using a powerful theorem from linear algebra. We run this algorithm on several examples, and find that the author and paper scores match closely with what is suggested by our intuition. The algorithm is theoretically sound and runs very fast in practice. We compare this index with several existing metrics and find that {\citex} gives far more accurate scores compared to the traditional metrics.
\end{abstract}

\section{Introduction}
\label{sec:Introduction}
In today's world, numerous papers are written by authors in many journals and conferences. It is difficult for people to judge the quality and impact of an author or a paper, even if they are experts, just by reading a few papers. Thus, measuring the relative importance of authors and papers published in scientific conferences and journals is very important. More importantly, there is a need for an index giving accurate results, which can be computed easily for a large collection of authors and papers.

Many metrics are available to evaluate the importance of journals like Impact factor \cite{G55}, Immediacy index, 
citation page rank \cite{PN76}, and Y-factor \cite{BSSL05}. 
There are many metrics which give scores to authors. 
Most famous of these are Hirsch's h-index \cite{H05}, Individual h-index \cite{BCK06}, Egghe's g-index \cite{E06}, and Zhang's e-index \cite{Z09}.  
However, till now, the only way to evaluate the impact of a paper, is to count the number of citations. 
It has been observed that surveys and review articles receive more citations than high quality original research papers. 
Self citations also increase the number of citations of a paper.

This is why we propose a new metric for evaluating papers for the first time. This new citation index {\citex}, gives normalized scores to authors and papers to determine their rankings. Using these scores, we can get an objective measure of the reputation of an author and the impact of a paper. Paper scores are calculated not solely based on the number of citations. 
Apart from giving scores to papers, we give scores to authors. The author scores and paper scores reinforce each other. 
Thus, an influential author will increase the score of the paper (s)he writes. 
An author's score increases if (s)he writes a good paper. 
Since the author score increases the score of a paper written by the author, it will be a general tendency to 
write a paper with an influential author.
To prevent this, we assign scores to authors in such a way that the score of a paper gets uniformly divided by the 
number of people who have co-authored the paper. This basic model can be very easily extended to weighted distribution of scores, 
where a first author who has the highest contribution receives more weightage than an author who has less contribution.

Existing literature rate an author based on the number of papers that he/she has written, the total number of citations received, 
average number of citations per paper etc. In Section \ref{sec:Related Work}, we discuss each of these metrics and state their 
advantages and disadvantages. 
None of the existing techniques take into account how the paper scores of an author influence the author's standing in the academic community, 
because the paper score is calculated solely based on the number of citations. 
Our {\citex} index is inspired by the ideas of PageRank \cite{page1999pagerank,brin1998anatomy,bollen2006journal} and HITS \cite{K99} algorithms for ranking web pages. In this scheme, paper scores are updated depending on author scores and author scores are updated based on paper scores.  This is often referred to as the \emph{Principle of Repeated Improvement} \cite{easley2010networks}. We prove that the scores asymptotically converge in the limit when the number of iterations is large. In practice, the scores converge within a few iterations.

The main idea is to consider the authors and papers as a network with a disjoint set of nodes, with the set of authors and the set of papers as vertices 
in the two sets. An edge exists between an author and a paper, if the author has written the corresponding paper. 
Apart from these, there is a citation subgraph, consisting of papers as the vertices. A directed edge exists from node $i$ to node $j$, 
if paper $i$ cites paper $j$. At each step, a paper score gets uniformly divided amongst all its authors. 
The paper score is the sum of the scores of the authors who have written the paper and the scores due to citation.

The currently available scores like h-index count the number of citations, but not the impact of these papers which have cited this paper. 
This can be easily manipulated by self-citations. Moreover, since many indices like h-index and number of citations are integers, it is difficult to distinguish between two authors or papers with the same value of the index. Our index has a higher \emph{discriminatory power}, since it is a real number between 0 and 1. It is highly unlikely that two authors or papers will have the same value for {\citex}.
There is also a non-uniformity across various disciplines. 
In sciences, there is a general tendency to work in large groups. These papers also receive large number of citations compared 
to papers in computer science. Since the score is divided uniformly among multiple authors, each author score will
be reduced, thus affecting paper scores as well. 
Our index also gives credit to authors who write single-author papers. This should not deter people to do collaborative research. 

The problem can be fine-tuned in a number of ways. We can take into consideration the recommendation of authors by other authors and
consider weighted distribution. The problem also has a number of ramifications. 
A similar index can be designed for product-customer recommendation system, where customers can recommend each other depending 
upon the reputations (similar to citation of papers), and a customer can recommend a product (similar to writing papers). 
The difference here is that, the score of a product is not uniformly divided amongst customers, and information cascades have to be taken into 
account when calculating the product scores. Other interdependent networks, which reinforce each other, can be treated similarly.

The rest of the paper is organized as follows. In section \ref{sec:Related Work}, we discuss the related work on citation analysis, define the metrics and compare them. In section \ref{sec:Problem definition}, we define the problem and propose the model to analyze it. We give an informal description of the algorithm and present the rules to iteratively compute the author and paper scores in section \ref{sec:Informal description}. In section \ref{sec:Mathematical Analysis}, we mathematically analyze the iterative procedure and prove that it converges to the eigenvectors of certain matrices. In section \ref{sec:Illustrative examples}, we execute the algorithm on some illustrative examples, and show that the author scores and paper scores give good indication of their importance, as can be seen from the underlying graph structure. In section \ref{sec:Extensions}, we discuss some extensions of the basic algorithm and future direction to work on. We conclude the paper in section \ref{sec:Conclusion} with some future directions to work on.
 
\section{Related work}
\label{sec:Related Work}
\subsection{Different metrics used in citation analysis}
Previously, there have been several attempts to measure the impact of authors and papers. We list here some of them along with their definition.

\begin{table*}
\caption{Comparison between different citation metrics}
\begin{center}
\label{tab:citation-metrics}
\begin{tabular}{|p{4cm}|p{5cm}|p{7.5cm}|}
\hline
\textbf{Citation Metric} & \textbf{Advantage} & \textbf{Disadvantage} \\
\hline
\emph{Number of papers} & Measure of productivity. & Importance of papers not considered. \\
\hline
\emph{Number of citations} & Measures impact of an author. & A few highly cited papers increase the total. \\
& & Survey and review articles are cited more than original research papers. \\
& & Favors established authors. \\
\hline
\emph{Average number of citations per paper} & Allows comparison of scientists of different ages. & Rewards low productivity. \\
\hline
\emph{Average number of citations per author} & Measures impact of an author. & Difficult to distinguish between authors whose average is same, but citation patterns are different. \\
\hline
\emph{Average number of papers per author} & Measures average productivity & Does not measure impact of papers. \\
\hline
\emph{Average number of authors per paper} & Measures collaboration between authors. & Does not consider importance of authors and papers. \\
\hline
\emph{h-index} & Measures both the quality and quantity of scientific output. & Does not account for the number of authors of a paper. \\
& & Different fields with different number of citations will have different h-index. \\
& & Can be manipulated through self-citations. \\
\hline
\emph{g-index} & Gives more weight to highly-cited articles. & Unlike the h-index, the g-index saturates whenever the average number of citations for all published papers exceeds the total number of published papers. \\
\hline
\emph{e-index} & Differentiates between scientists with identical h-indices but different citations. & Can't be used independently. Must be used together with the h-index. \\
\hline
\emph{Number of papers with at least $c$ citations} & Measures the broad and sustained impact of an author. & Difficult to find the right value of $c$. Different values of $c$ favors different authors. \\
\hline
\emph{Number of citations to the $k$ most cited papers} & Identifies the most influential authors. & Not a single number, so difficult to compare. Different values of $k$ favors different authors. \\
\hline
\emph{Eigenfactor} & Takes into account impact of the citing papers in addition to the number of citations. & Does not give author scores. Importance of authors have to inferred indirectly from the papers (s)he has written. \\
\hline
\end{tabular}
\end{center}
\end{table*}

\begin{enumerate}
	\item \textbf{Number of papers ($N_p$):} Total number of papers written by an author.
	\item \textbf{Number of citations ($N_c$):} Total number of citations for all papers written by an author.
	\item \textbf{Average number of citations per paper:} Ratio of total number of citations and total number of papers, \emph{i.e.,} $\frac{N_c}{N_p}$. This is sometimes also called the \emph{impact factor}.
	\item \textbf{Average number of citations per author:} For each paper, its citation count is divided by the number of authors for that paper to give the normalized citation count for the paper. The normalized citation counts are then summed across all papers to give the average number of citations per author.
	\item \textbf{Average number of papers per author:} For each paper, the inverse of the number of authors gives the normalized author count for the paper. The normalized author counts are then summed across all papers to give the average number of papers per author.
	\item \textbf{Average number of authors per paper:} The sum of the author counts across all papers, divided by the total number of papers.
	\item \textbf{$h$-index \cite{H05}:} An author has index $h$, if $h$ of his $N_p$ papers have at least $h$ citations each, and the rest of the $N_p - h$ papers have no more than $h$ citations each.
	\item \textbf{$g$-index \cite{E06}:} Given a set of articles ranked in decreasing order of the number of citations that they received, the $g$-index is the (unique) largest number such that the top $g$ articles together received at least $g^2$ citations.
	\item \textbf{$e$-index \cite{Z09}:} It is the square root of surplus citations in the $h$-set beyond the theoretical minimum ($h^2$) required to obtain a $h$-index of $h$. It is useful for highly cited scientists and for comparing those with the same $h$-index but different citation patterns.
	\item \textbf{Number of significant papers:} Total number of papers with more than $c$ citations for some integer $c$.
	\item \textbf{Number of citations to the most cited papers:} Total number of citations to the $k$ most cited papers for some integer $k$.
	\item \textbf{Eigenfactor:} The Eigenfactor score \cite{bergstrom2008eigenfactor} is a rating of the total importance of a scientific journal. Journals are rated according to the number of incoming citations, with citations from highly ranked journals weighted to make a larger contribution to the Eigenfactor than those from poorly ranked journals \cite{bergstrom2007measuring}. As a measure of importance, the Eigenfactor score scales with the total impact of a journal. Journals generating higher impact to the field have larger Eigenfactor scores. However, it is not clear whether Eigenfactor gives better estimate than raw citation count. \cite{davis2008eigenfactor}
\end{enumerate}

\subsection{Comparison between different citation metrics}
In this section, we compare the different citation indices and state their advantages and disadvantages. The comparison is presented in Table \ref{tab:citation-metrics}.

\subsection{Comparison with similar works}
There are a number of previous attempts to rank authors and papers based on importance. The \textsc{SimRank} algorithm by Jeh and Widom \cite{jeh2002simrank} gives a measure of the similarity between two objects based on their relationships with other objects. Their basic idea is that two objects are similar if they are related to similar objects. Note that this only measures the similarity of two objects, not their relative ranking, so this is different from what we are trying to do. Zhou et. al. \cite{zhou2007co} proposed a method for co-ranking authors and their publications using several networks associated with authors and papers. Although there is some similarity between our algorithm and their approach, there are fundamental differences between the two. Their co-ranking framework is based on coupling two random walks that separately rank authors and documents using the \textsc{PageRank} algorithm. Our algorithm is designed from scratch and does not use the \textsc{PageRank} algorithm. Moreover, our algorithm is much simpler and the computations required is also far lesser than what is required in their method. Walker et. al. \cite{walker2007ranking} gave a new algorithm called \textsc{CiteRank}. The ranking of papers is based on a network traffic model, which uses a variation of the \textsc{PageRank} algorithm. A paper is selected randomly from the set of all papers with a probability that decays exponentially with the age of the paper. Chen et. al. \cite{chen2007finding} uses a \textsc{PageRank} based algorithm to assess the relative importance of all publications. Their goal is to find some exceptional papers or ``gems" that are universally familiar to physicists. Sun and Giles \cite{sun2007popularity} propose a popularity weighted ranking algorithm for academic digital libraries. They use the popularity of a publication venue and compare their method with the \textsc{PageRank} algorithm, citation counts and the HITS algorithm.

\subsection{Some structures in citation analysis}
\begin{itemize}
	\item \textbf{Collaboration graph:} This is a graph associated with the authors. The nodes of the graph are the authors. There is an undirected edge between two nodes, if the corresponding authors have written a paper together.
	\item \textbf{Citation graph:} This is a graph associated with the papers. The nodes of the graph are the papers. There is a directed edge from a paper to another paper, if the first paper has cited the second paper.
	\item \textbf{Publication graph:} This is a graph relating the authors with the papers. The nodes of the graph are the authors and the papers. There is an undirected edge between two nodes, if the author has written the paper.
\end{itemize}

\section{Problem definition and model}
\label{sec:Problem definition}
We have a set of $m$ \emph{authors} $\cA = \{a_1,\ldots,a_m\}$ and a set of $n$ \emph{papers} $\cP = \{p_1,\ldots,p_n\}$. We represent this by a \emph{publication graph} $G_P = (V_P, E_P)$, whose vertices are the set of authors and papers, \emph{i.e.}, $V_P = \cA \cup \cP$. There is an undirected edge between author $a_i$ and paper $p_j$, if author $a_i$ has written paper $p_j$. Note that this is a \emph{symmetric} relation, so the edges are \emph{undirected}. Since there are only edges between authors and papers, the publication graph is a \emph{bipartite graph}. Associated with this, there is an $m \times n$ \emph{publication matrix} $M$, whose rows and columns are $a_1,\ldots,a_m$ and $p_1,\ldots,p_n$ respectively, and whose $(i,j)^{th}$ entry $m_{ij} = 1$, if and only if author $a_i$ has written paper $p_j$.

Moreover, there is a \emph{citation graph} $G_C = (V_C, E_C)$ associated with the papers, whose vertices are the set of papers, \emph{i.e.}, $V_C = \cP$. There is a directed edge from paper $p_j$ to paper $p_k$, if paper $p_j$ has cited paper $p_k$. Note that this is an \emph{asymmetric} relation, so the edges are \emph{directed}. Associated with this, there is an $n \times n$ \emph{citation matrix} $C$, whose both rows and columns are $p_1,\ldots,p_n$, and whose $(j,k)^{th}$ entry $c_{jk} = 1$, if and only if paper $p_j$ has cited paper $p_k$. Note that the citation graph can't have any \emph{directed cycle}. This is because a paper can only cite a previously published paper, so they are \emph{totally ordered} in time. This also means that if the papers are numbered in \emph{decreasing order of time} (newer first), the resulting citation matrix will be \emph{upper-triangular}. An example of a publication graph and a citation graph is given in Figure \ref{citation-graph}.

\begin{figure}[ht]
\begin{center}
\includegraphics[width=2in]{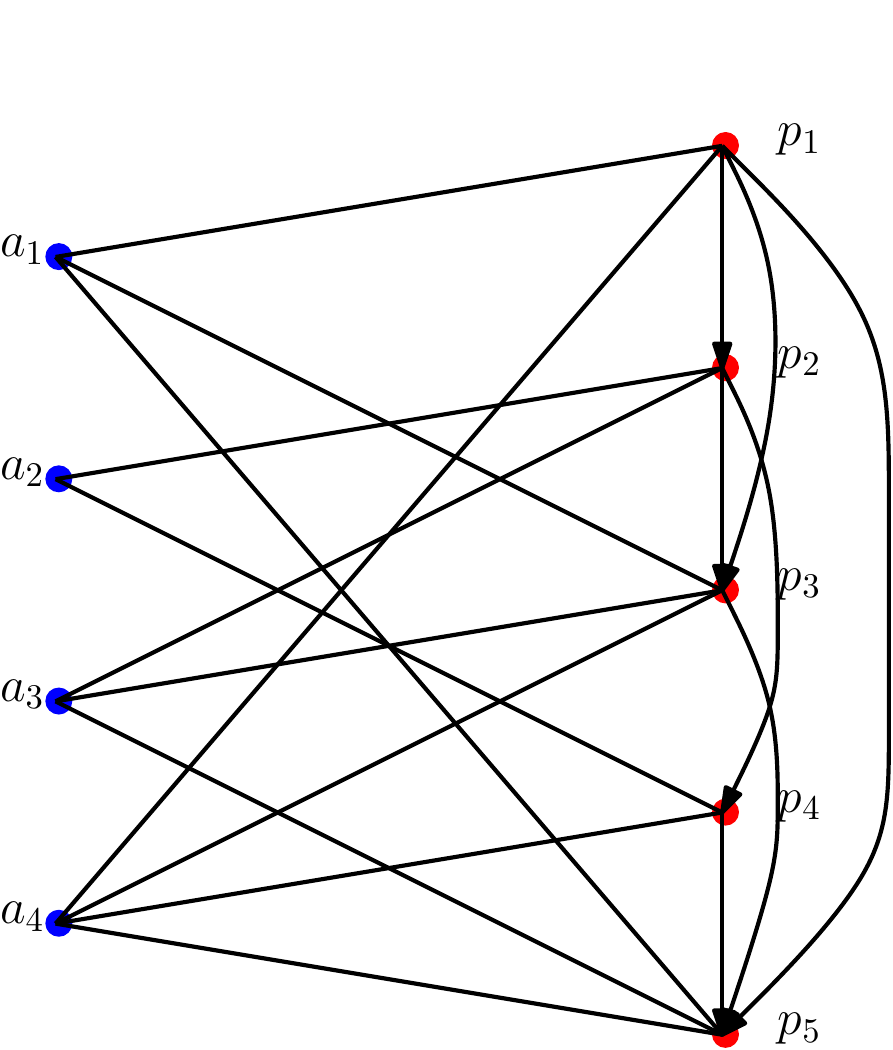}
\caption{The publication and citation graphs for Example 1 showing authors, papers and citations.}
\label{citation-graph}
\end{center}
\end{figure}
 
The following sets are important for further development.
\begin{enumerate}
	\item For an author $a$, $PAPERS(a)$ is defined as the set of papers written by author $a$. In other words, $PAPERS(a) = \{p \in \cP: (a,p) \in E_P\}$.
	\item For a paper $p$, $AUTHORS(p)$ is defined as the set of authors who have written paper $p$. In other words, $AUTHORS(p) = \{a \in \cA: (a,p) \in E_P\}$.
	\item For a paper $p$, $CITE(p)$ is defined as the set of papers who have cited paper $p$. In other words, $CITE(p) = \{q \in \cP: (q,p) \in E_C\}$.
	\item For a paper $p$, $REF(p)$ is defined as the set of papers which have been given as reference (cited) by paper $p$. In other words, $REF(p) = \{q \in \cP: (p,q) \in E_C\}$.
\end{enumerate}

Our goal is to assign scores to authors and papers using the structure of the publication and citation graphs, so that important authors and papers get higher scores.

\section{Description of the Citex index}
\label{sec:Informal description}
\subsection{Informal description of the algorithm}
In this section, we give an overview of our proposed algorithm. The algorithm maintains a set of author scores and paper scores, which are initially set to 1. This initial choice of scores is arbitrary, and the scores can be set to any nonzero value. Then we update the scores considering the relationship between authors and papers (who has authored which paper) and relationship between papers (which paper has cited which paper). This critically uses the publication graph and the citation graph. We use the \emph{Principle of Repeated Improvement} \cite{easley2010networks} to iteratively compute the new scores based on the previous scores. More specifically, the author scores for the next iteration is computed from the paper scores for the current iteration. The paper scores for the next iteration is computed from the author scores and the paper scores for the current iteration. In every iteration, we normalize the scores by dividing them by the sum of the individual scores, so that each of them lies between 0 and 1, and they add up to 1. We continue to do this till the author scores and the paper scores converge or a specified number of iterations have been completed. The \emph{Principle of Repeated Improvement} states that each improvement of author scores will lead to a further improvement of paper scores, and vice versa. The final author scores and paper scores are the measure of importance of the authors and the papers. The higher the score is, the higher is the impact of an author and a paper.

\subsection{Computing author and paper scores}
\label{sec:Computing scores}
For each author $a_i$, we have an \emph{author score} (\emph{a-score}) $x_i$, and for each paper $p_j$, we have a \emph{paper score} (\emph{p-score}) $y_j$. We represent the set of author scores as a column vector $\xx = (x_1,\ldots,x_m)^T$ and the set of paper scores as a column vector $\yy = (y_1,\ldots,y_n)^T$. We initialize all author and paper scores to one, \emph{i.e.}, $\xx = \yy = \mathbf{1}$. Then, we iteratively update the $a$-scores and $p$-scores using the following rules.

\begin{enumerate}
	\item For each paper $p_j$, its \emph{adjusted $p$-score} $\bar{y}_j$ is given by the $p$-score $y_j$ divided by the number of authors who have written the paper. In other words, $\bar{y}_j = \frac{y_j}{k}$, for $j=1,\ldots,n$, where $k = |AUTHORS(p_j)|$ is the number of co-authors of the paper $p_j$.
	\item For each author $a_i$, set his $a$-score $x_i$ to be the sum of the adjusted $p$-scores of all the papers that he has authored. In other words, $x_i = \sum_{j \in PAPERS(i)} \bar{y}_j$, for $i=1,\ldots,m$.
	\item For each paper $p_j$, set its $p$-score $y_j$ to be the sum of the $a$-scores of all the authors who have co-authored the paper $p_j$. In other words, $y_j = \sum_{i \in AUTHORS(j)} x_i$, for $j=1,\ldots,n$.
	\item For each paper $p_j$, add to its $p$-score $y_j$, the sum of the $p$-scores of all the papers who have cited the paper $p_j$. In other words, $y_j = y_j + \sum_{k \in CITE(j)} y_k$, for $j=1,\ldots,n$.
\end{enumerate}

We normalize the scores by dividing the author (paper) scores by the sum of the author (paper) scores, so that each score lies between 0 and 1, and the sum of the scores is 1.

\section{Mathematical analysis}
\label{sec:Mathematical Analysis}
\subsection{Analysis of author scores and paper scores}
We observe that the rule $y_j = \sum_{i \in AUTHORS(j)} x_i$ can be rewritten as $y_j = \sum_{i=1}^m m_{ij} x_i$, since $m_{ij} = 1$ if and only if $i \in AUTHORS(j)$. Consider the matrix-vector equation $\yy \leftarrow M^T \xx$. The $j^{th}$ row of this equation is $y_j = \sum_{i=1}^m m_{ij} x_i$. Hence, this matrix-vector equation succinctly encodes all $n$ scalar equations for $j=1,\ldots,n$.

The corresponding equation for $\xx$ is similar, but a little more involved. The equation $x_i = \sum_{j \in PAPERS(i)} \bar{y}_j$ can be written as $x_i = \sum_{j=1}^n m_{ij} \bar{y}_j$, since $m_{ij} = 1$ if and only if $j \in PAPERS(i)$. Let $\bar{\yy} = (\bar{y}_1,\ldots,\bar{y}_n)^T$. Consider the matrix-vector equation $\xx \leftarrow M \bar{\yy}$. The $i^{th}$ row of this equation is $x_i = \sum_{j=1}^n m_{ij} \bar{y}_j$. Hence, this matrix-vector equation succinctly encodes all $m$ scalar equations for $i=1,\ldots,m$. Further note that, $\bar{y}_j = \frac{y_j}{|AUTHORS(p_j)|} = \frac{y_j}{\sum_{i=1}^m m_{ij}}$. Hence, $x_i = \sum_{j=1}^n \left(\frac{m_{ij}}{\sum_{i=1}^m m_{ij}}\right) y_j = \sum_{j=1}^n w_{ij} y_j$, where $w_{ij} = \frac{m_{ij}}{\sum_{i=1}^m m_{ij}}$ is the \emph{weight} associated with the paper $p_j$. Now, $\xx$ can be written as $\xx \leftarrow W \yy$, where $W$ is the $m \times n$ \emph{weight matrix} whose $(i,j)^{th}$ entry is $w_{ij}$.

The equation $y_j = y_j + \sum_{k \in CITE(j)} y_k$ can be rewritten as $y_j = y_j + \sum_{k=1}^n c_{kj} y_k$, since $c_{kj} = 1$ if and only if $k \in CITE(j)$. This can be written as the matrix-vector equation $\yy \leftarrow (I + C^T) \yy$, where $I$ is the $n \times n$ \emph{identity matrix}.

Let the initial author vector and paper vector be $\xx^{\langle 0\rangle}$ and  $\yy^{\langle 0\rangle}$ respectively. If we start with the equation $\xx \leftarrow W \yy$, the successive iterations proceed as below.
\begin{align}
\xx^{\langle 1\rangle} &= W \yy^{\langle 0\rangle}, \\
\yy^{\langle 1\rangle} &= M^T \xx^{\langle 1\rangle} = M^T W \yy^{\langle 0\rangle}, \\
\yy^{\langle 1\rangle} &= (I + C^T) \yy^{\langle 1\rangle} = (I + C^T) M^T W \yy^{\langle 0\rangle}.
\end{align}

Similarly, if we start with the equation $\yy \leftarrow M^T \xx$, the successive iterations proceed as below.
\begin{align}
\yy^{\langle 1\rangle} &= M^T \xx^{\langle 0\rangle}, \\
\yy^{\langle 1\rangle} &= (I + C^T) \yy^{\langle 1\rangle} = (I + C^T) M^T \xx^{\langle 0\rangle}, \\
\xx^{\langle 1\rangle} &= W \yy^{\langle 1\rangle} = W (I + C^T) M^T \xx^{\langle 0\rangle}.
\end{align}

Proceeding similarly, at the $k$-th iteration the author and paper vectors are given by, 
\begin{align}
\xx^{\langle k\rangle} &= [W (I + C^T) M^T]^k \xx^{\langle 0\rangle}, \label{x-iter} \\
\yy^{\langle k\rangle} &= [(I + C^T) M^T W]^k \yy^{\langle 0\rangle}. \label{y-iter}
\end{align}

\subsection{Proof of convergence of author scores and paper scores}
In this section, we will prove the following theorem.
\begin{theorem}
The sequences $\xx^{\langle k\rangle}$ and $\yy^{\langle k\rangle}$, $k = 0, 1, 2, \ldots$ converge to the limits $\xx^\star$ and $\yy^\star$ respectively. Moreover, $\xx^\star$ is the principal eigenvector of the matrix $W (I + C^T) M^T$ and $\yy^\star$ is the principal eigenvector of the matrix $(I + C^T) M^T W$. Further, both $\xx^\star$ and $\yy^\star$ are non-negative and non-zero vectors.
\end{theorem}

\begin{proof}
From the above discussion we have,
$\xx^{\langle k\rangle} = P^k \xx^{\langle 0\rangle}$ and $\yy^{\langle k\rangle} = Q^k \yy^{\langle 0\rangle}$, where $P = W (I + C^T) M^T$ and $Q = (I + C^T) M^T W$. Note that $P$ is an $m \times m$ square matrix, whereas $Q$ is an $n \times n$ square matrix. Moreover, $\xx^{\langle k+1\rangle} = P^{k+1} \xx^{\langle 0\rangle} = P \cdot P^k \xx^{\langle 0\rangle} = P \xx^{\langle k\rangle}$.
If the author score $\xx^{\langle k\rangle}$ converges to the vector $\xx^{\star}$ in the limit when $k \rightarrow \infty$, then this vector should satisfy $P \xx^{\star} = \xx^{\star}$. 
This means that $\xx^{\star}$ is an eigenvector of $P$, with the corresponding eigenvalue being 1. Similarly, if the paper score $\yy^{\langle k\rangle}$ converges to the vector $\yy^{\star}$ in the limit when $k \rightarrow \infty$, then $\yy^{\star}$ must be an eigenvector of $Q$, with the corresponding eigenvalue being 1.

To prove that a non-negative eigenvalue and a non-negative eigenvector exists, we use the following theorem from linear algebra.

\begin{theorem}[\textsc{Perron-Frobenius Theorem}] \cite{perron1907theorie, frobenius1912matrizen, easley2010networks}
Let $A = (a_{ij})$ be an $n \times n$ non-negative matrix, meaning that $a_{ij} \geq 0, \forall i,j: 1\leq i,j \leq n$. Then the following statements hold.
\begin{enumerate}
\item $A$ has a real eigenvalue $c \geq 0$ such that $c > |c'|$ for all other eigenvalues $c'$.
\item There is an eigenvector $v$ with non-negative real components corresponding to the largest
eigenvalue $c: Av = cv, v_i \ge 0, 1 \leq i \leq n$, and $v$ is unique up to multiplication by a constant.
\item If the largest eigenvalue $c$ is equal to $1$, then for any starting vector $\xx^{\langle 0\rangle} \neq 0$ with non-negative components, the sequence of vectors $A^k \xx^{\langle 0\rangle}$ converge to a vector in the direction of $v$ as $k \rightarrow \infty$.
\end{enumerate}
\end{theorem}

Thus, by the Perron-Frobenius theorem, the author and paper scores both converge to unique non-negative vectors $\xx^{\star}$ and $\yy^{\star}$, after repeated applications of the update rules. These two vectors are the limiting values of the author and paper scores. Moreover, none of the vectors $\xx^{\star}$ and $\yy^{\star}$ can be the zero vector. At least one of their components must be non-zero, because the initial vectors $\xx^{\langle 0\rangle}$ and $\yy^{\langle 0\rangle}$ are the all-1 vectors, and at each iteration the vectors are normalized. So the sum of their components add up to 1.
\end{proof}

\subsection{Time-complexity of each iteration}
We have to multiply some matrices as can be seen from equations (\ref{x-iter}) and (\ref{y-iter}). Computing the product of a $m \times n$ matrix and a $n \times p$ matrix requires $O(mnp)$ time. Computing the matrices $W (I + C^T) M^T$ and $(I + C^T) M^T W$ takes $O(mn(m+n))$ time each. Hence, each iteration can be done in $O(mn(m+n))$ time.

\section{Experimental analysis}
\label{sec:Illustrative examples}
We compute the author and paper scores for some graphs and show that they match with our intuition. 

\subsection{Example 1}
\label{sec:Example 1}
For the graph in Figure \ref{citation-graph}, the author and paper vectors, publication matrix and citation matrix are given below. Note that for this example, $m = 4, n = 5$.

\begin{equation*}
\xx = [1, 1, 1, 1], \yy = [1, 1, 1, 1, 1],
\end{equation*}
\begin{equation*}
M =
\begin{bmatrix}
1 & 0 & 1 & 0 & 1 \\
0 & 1 & 0 & 1 & 0 \\
0 & 1 & 1 & 0 & 1 \\
1 & 0 & 1 & 1 & 1
\end{bmatrix}
, C =
\begin{bmatrix}
0 & 1 & 1 & 0 & 1 \\
0 & 0 & 1 & 1 & 0 \\
0 & 0 & 0 & 0 & 1 \\
0 & 0 & 0 & 0 & 1 \\
0 & 0 & 0 & 0 & 0
\end{bmatrix}
.
\end{equation*}

The author scores after 10 iterations are given below. These scores converge (up to 3 decimal places) as there is no change between two successive iterations.
\begin{align*}
\xx &= [0.259,0.132,0.289,0.320], \\
\yy &= [0.082,0.141,0.264,0.123,0.390].
\end{align*}

Intuitively it is clear that author $a_4$ and paper $p_5$ should get the highest scores. Author $a_4$ has written 4 papers $p_1,p_3,p_4,p_5$ and some of them have high paper scores. Similarly, $p_5$ has been written by 3 authors $a_1,a_3,a_4$ and cited by 3 papers $p_1,p_3,p_4$, some of which have high scores. $a_2$ gets the lowest score as it has written only 2 papers $p_2,p_4$, none of which have high paper scores. Similarly, $p_1$ gets the lowest score since it has no citations, although it has two authors $a_1,a_4$.

\begin{figure}[ht]
\begin{center}
\includegraphics[width=2in]{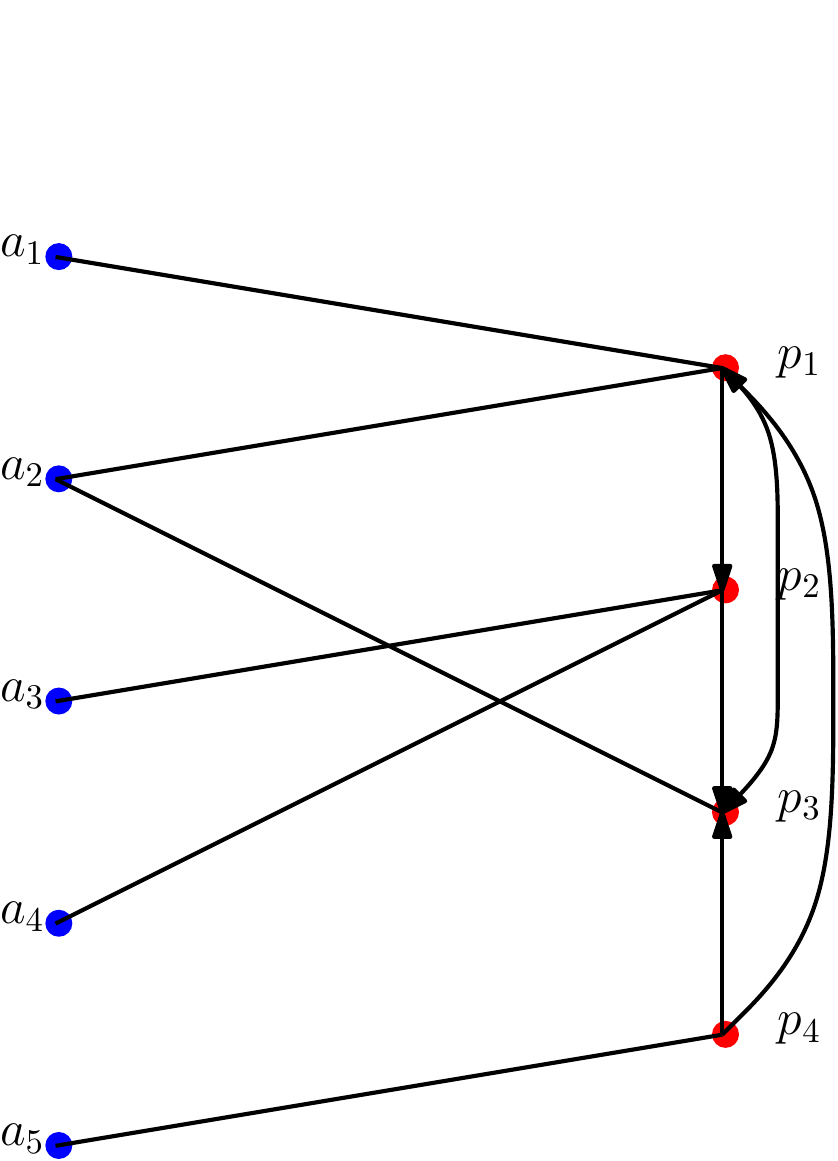}
\caption{The publication and citation graphs for Example 2 showing authors, papers and citations.}
\label{citation-graph2}
\end{center}
\end{figure}
 
\subsection{Example 2}
\label{sec:Example 2}
For the graph in Figure \ref{citation-graph2}, the author and paper vectors, publication matrix and citation matrix are given below. Note that for this example, $m = 5, n = 4$.

\begin{equation*}
\xx = [1, 1, 1, 1, 1], \yy = [1, 1, 1, 1],
\end{equation*}
\begin{equation*}
M =
\begin{bmatrix}
1 & 0 & 0 & 0 \\
1 & 0 & 1 & 0 \\
0 & 1 & 0 & 0 \\
0 & 1 & 0 & 0 \\
0 & 0 & 0 & 1
\end{bmatrix}
, C =
\begin{bmatrix}
0 & 1 & 1 & 0 \\
0 & 0 & 1 & 0 \\
0 & 0 & 0 & 0 \\
1 & 0 & 1 & 0
\end{bmatrix}
.
\end{equation*}

The author scores after 10 iterations are given below. These scores converge (up to 3 decimal places) as there is no change between two successive iterations.
\begin{align*}
\xx &= [0.106,0.590,0.152,0.152,0.000], \\
\yy &= [0.212,0.304,0.484,0.000].
\end{align*}

Intuitively it is clear that author $a_2$ and paper $p_3$ should get the highest scores. Author $a_2$ has written two papers $p_1$ and $p_3$. In turn, $p_3$ has been cited by papers $p_1, p_2$ and $p_4$. On the other hand, $p_4$ gets the lowest paper score 0, since it is not cited by any paper. $a_5$ gets the lowest author score since it has only written paper $p_4$, which has a score 0. These scores matches with our intuition.

\section{Extensions}
\label{sec:Extensions}
\subsection{Weights on edges of the publication graph}
In Section \ref{sec:Mathematical Analysis} we assumed that all authors contribute equally to a paper. 
However, this is not true in practice. Different co-authors have different contribution to a paper. 
We can easily incorporate this feature in our {\citex} index, by slightly modifying the publication matrix $M$.
In Section \ref{sec:Mathematical Analysis}, $M$ was a 0-1 matrix. 
For the weighted version, we construct a matrix $N$, where edge weight $n_{ij}$ denotes the contribution of author $i$ in writing paper $j$. 
An author having higher contribution has higher weight, compared to an author having lesser contribution. 
The new matrix $N$  might not be a 0-1 matrix. 
$W'$ is the matrix corresponding to $W$. 
Hence, $x_i = \sum_{j=1}^n \left(\frac{n_{ij}}{\sum_{i=1}^m n_{ij}}\right) y_j = \sum_{j=1}^n w'_{ij} y_j$, where $w'_{ij} = \frac{n_{ij}}{\sum_{i=1}^m n_{ij}}$ is the weight associated with the paper $p_j$. Now, $\xx$ can be written as $\xx \leftarrow W' \yy$, where $W'$ is the weight matrix whose $(i,j)^{th}$ entry is $w'_{ij}$.

Now the author and paper vectors at the $k$-th iteration can be written as,
\begin{align*}
\xx^{\langle k\rangle} &= [W' (I + C^T) N^T]^k \xx^{\langle 0\rangle}, \\
\yy^{\langle k\rangle} &= [(I + C^T) N^T W']^k \yy^{\langle 0\rangle}.
\end{align*}

\subsection{Reputation of authors}
Our {\citex} index can be modified to include the reputation of authors. Each author can rank other authors who he/she believes has done original, ground breaking work.
We define the \emph{author reputation graph} $G_R = (V_R, E_R)$, similar to the citation graph. The vertices are the set of authors, \emph{i.e.}, $V_R = \cA$. Thus, the number of vertices in this graph is $m$.  
A directed edge from node $i$ to node $j$ of weight $r_{ij}$ exists, if author $i$ has rated author $j$ with a score $r_{ij}$. 
The \emph{author reputation matrix} $R$ is defined similarly.   
For an author $a$, $REP(a)$ is defined as the set of authors who have ranked author $a$. 
In other words, $REP(a) = \{b \in \cA: (b,a) \in E_R\}$.

On incorporating the reputation matrix, another rule is added to the list. 
For each author $a_i$, add to its $a$-score $x_i$, the sum of the $a$-scores of all the authors who have rated author $a_i$. In other words, $x_i = x_i + \sum_{k \in REP(i)} x_k = x_i + \sum_{k=1}^m r_{ki} x_k$, for $i=1,\ldots,m$.

Let the initial author vector and paper vector be $\xx^{\langle 0\rangle}$ and  $\yy^{\langle 0\rangle}$ respectively. If we start with the equation $\xx \leftarrow W \yy$, the successive iterations proceed as below.
\begin{align*}
\xx^{\langle 1\rangle} &= W \yy^{\langle 0\rangle}, \\
\xx^{\langle 1\rangle} &= (I + R^T) \xx^{\langle 1\rangle} = (I + R^T) W \yy^{\langle 0\rangle}\\
\yy^{\langle 1\rangle} &= M^T \xx^{\langle 1\rangle} = M^T (I + R^T) W \yy^{\langle 0\rangle}, \\
\yy^{\langle 1\rangle} &= (I + C^T) \yy^{\langle 1\rangle} = (I + C^T) M^T (I + R^T) W \yy^{\langle 0\rangle}.
\end{align*}

Proceeding similarly, at the $k$-th iteration the author and paper vectors are given by, 
\begin{align*}
\xx^{\langle k\rangle} &= [(I + R^T) W (I + C^T) M^T]^k \xx^{\langle 0\rangle}, \\
\yy^{\langle k\rangle} &= [(I + C^T) M^T (I + R^T) W]^k \yy^{\langle 0\rangle}.
\end{align*}

\subsection{Evaluation of journals and conferences}
\label{sec:Journals}
One important question is how to measure the quality of scientific journals and conferences. If we have the author scores and paper scores of all authors and papers published in a journal/conference, we can use the average author score and the average paper score as a metric for determining the quality of the journal/conference.

\section{Conclusion and future directions}
\label{sec:Conclusion}
In this paper, we have proposed a new citation metric {\citex} to judge the quality of authors and papers in scientific publications. {\citex} assigns scores to authors and papers, with higher scores indicating more importance, by analyzing the link structures in the publication graph and the citation graph. We also considered some extensions to the basic scheme. Here are some future directions to work on.
\begin{itemize}
	\item In a real-world scenario, authors and papers will be added over time. Dynamically modifying the scores from the current scores in an incremental fashion is a challenging problem.
	\item Applying this framework to the setting of customer recommendation of products will be an interesting idea. Here the nodes of the graph are customers and products. A customer can give a score to a product, which is like the weighted version of the publication graph.
	\item This is an example of an \emph{interdependent network}, where there are two graphs -- the collaboration/reputation graph and the citation graph. In addition, there is a publication graph which records the cross-edges between the nodes in the two graphs. Extending {\citex} to other interdependent networks will be an interesting direction to think about.
	\item Using further parameters such as time of publication and age of authors, in addition to the link structure will require further thoughts.
	\item Can this technique be extended to \emph{directly} assign scores to journals and conferences, rather than doing it \emph{indirectly}, as in section \ref{sec:Journals}?
	\item The analytical power of eigenvector-based methods is not yet fully understood. It would be interesting to pursue this question in the context of the algorithm presented here. Considering random graph models that contain enough structure to capture certain global properties of the model is a promising direction.
\end{itemize}

\bibliographystyle{plain}
\bibliography{Citex}

\end{document}